\documentclass[aps,pra,superscriptaddress,nofootinbib,twocolumn]{revtex4-2}

\usepackage{amsmath,amsthm,amsfonts,amssymb,amscd} 
\usepackage[utf8]{inputenc}
\usepackage{indentfirst}
\usepackage{graphicx}
\usepackage{color}
\usepackage[T1]{fontenc}
\usepackage[american,ngerman,greek,brazilian,british]{babel}
\usepackage{float}
\usepackage[dvipsnames]{xcolor}
\usepackage[colorlinks=true,citecolor=Mulberry,linkcolor=prettygreen,urlcolor=MidnightBlue,hyperindex,breaklinks]{hyperref}
\usepackage{breakurl}
\usepackage{cleveref}
\usepackage{mathtools,xfrac}
\usepackage{bm}
\usepackage{ulem}
\usepackage{enumerate}
\usepackage{sidecap}
\usepackage{microtype}
\usepackage{enumitem}
\usepackage{bbold}
\usepackage{listings}
\usepackage{algorithm}
\usepackage{algpseudocode}
\usepackage{multirow}
\usepackage{setspace}
\usepackage{physics}
\usepackage{dsfont}
\usepackage{soul}
\usepackage{lipsum}  
\usepackage{comment}

\definecolor{awesome}{rgb}{0.93, 0.53, 0.18}
\definecolor{prettygreen}{RGB}{5,125,143}

\newcommand{\ot}{\otimes}

\newtheorem{theorem}{Theorem}
\newtheorem*{theorem*}{Theorem}
\newtheorem{definition}{Definition}
\newtheorem{observation}{Observation}

\newtheorem{lemma}{Lemma}

\begin{document}

\title{The power of quantum catalytic local operations}

\author{Patryk Lipka-Bartosik}
\address{Department of Applied Physics, University of Geneva, Geneva, Switzerland}
\affiliation{Center for Theoretical Physics, Polish Academy of Sciences, Warsaw, Poland}
\address{Institute of Theoretical Physics, Jagiellonian University, 30-348 Kraków, Poland}

\author{Jessica Bavaresco}
\address{Department of Applied Physics, University of Geneva, Geneva, Switzerland}
\address{Sorbonne Universit\'e, CNRS, LIP6, F-75005 Paris, France}

\author{Nicolas Brunner}
\address{Department of Applied Physics, University of Geneva, Geneva, Switzerland}

\author{Pavel Sekatski}
\address{Department of Applied Physics, University of Geneva, Geneva, Switzerland}

\begin{abstract}
A key result in entanglement theory is that the addition of a catalyst dramatically enlarges the set of possible state transformations via local operations and classical communication (LOCC). However, it remains unclear what is the interplay between classical communication and quantum catalysis. Here our aim is to disentangle the effect of the catalyst from that of classical communication. To do so, we explore a class of state transformations termed catalytic local operations (CLO) and compare it to LOCC and to stochastic LOCC augmented by bounded quantum communication. We show that these classes are incomparable and capture different facets of quantum state transformations. 
\end{abstract}

\maketitle

\section{Introduction}

Correlations between separated systems represent a fundamental concept in many scientific disciplines and a resource in various tasks. To understand correlations from an information-theoretic perspective one studies how they can be established and transformed under various operational restrictions, termed transformation classes. Arguably, local operations (LO) represent the minimal such class of operations: applied separately on different systems, they preserve their independence and are thus expected not to increase measures of correlations based on information as, e.g., captured by various data processing inequalities. 

In quantum information theory these ideas have been structured in the resource-theoretic approach~\cite{Chitambar_2019}, explicitly specifying which operations (e.g. LO) and which state (e.g. product states) are free, i.e. can be realized without consuming any resource, and, in contrast, which ones are resourceful (e.g. entangling gates or entangled states). Correlation measures must then be monotones under those operations.
In this context, a central question is to characterize genuinely quantum forms of correlations, and separate them from their classical counterpart (see e.g. Refs.~\cite{RevModPhys.81.865,Brunner_2014,ModiReview,uolaReview}). The most prominent example is entanglement theory, where LO are considered together with additional classical communication between the parties, resulting in the LOCC transformation class~\cite{Chitambar_2014}. Here, the addition of unbounded classical communication in particular guarantees that any amount of classical correlations comes for free.

\begin{figure}[b!]
    \centering
    \includegraphics[width=0.75\linewidth]{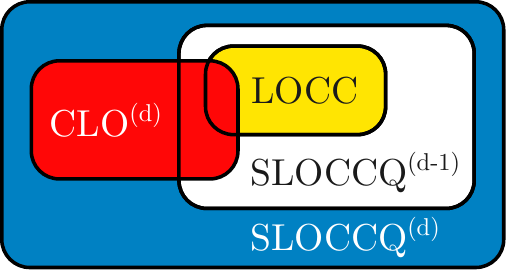}
    \caption{We consider three operational frameworks: catalytic local operations with bounded Schmidt number of the catalyst, $\text{CLO}^{(d)}$; local operations and classical communication, LOCC; and stochastic local operations with bounded quantum communication, $\text{SLOCCQ}^{(d)}$. By construction, every $\text{CLO}^{(d)}$ protocol can be simulated within $\text{SLOCCQ}^{(d)}$. However, our results demonstrate strict separations: there are transformations in $\text{CLO}^{(d)}$ that cannot be realized by $\text{SLOCCQ}^{(d-1)}$, and conversely, some tasks possible with LOCC cannot be reproduced by catalytic protocols, regardless of the catalyst’s dimension. Thus, these three classes capture genuinely different aspects of bipartite quantum state manipulation.}
    \label{fig:clod}
\end{figure}

Remarkably, it was discovered that the set of state transformations possible with LOCC is drastically enlarged by allowing an extra system to act as a catalyst \cite{Jonathan_1999,bennett2000exact,eisert2000catalysis,duan2005multiple,kondra2021catalytic,Rubboli_2022,review1,review2}. Quantum catalysis, which has also been studied in different contexts \cite{bavaresco2025catalytic,Wilming2017,Henao2023,Shiraishi2021,Brandao2015second}, shows that having access to an additional entangled state (a \textit{catalyst}), which must be returned in its initial state, can enable transformations that would otherwise be impossible. Conceptually, this reveals a novel way of exploiting entanglement: rather than acting as a consumable, it can also serve as a reusable resource. This also makes quantum catalysis a natural extension of the LOCC framework. In particular, this effect has operational consequences; for example, catalysis can improve teleportation~\cite{PhysRevLett.127.080502}, state merging \cite{kondra2021catalytic}, and entanglement distillation~\cite{kondra2021catalytic}.

This raises a number of fundamental questions: what is the real power of the catalyst in enabling state transformations? In particular, what state transformations are possible in a setting where only LO and catalysis are considered, without involving classical communication? 

In order to disentangle the effect of the catalyst from that of classical communication, we explore a class of state transformations termed catalytic local operations (CLO). Here, distant parties share a system and a catalyst to which they apply local operations (represented by completely positive trace-preserving (CPTP) maps) and must return the catalyst exactly unchanged. Yet, the system and catalyst can become correlated at the end of the process. Importantly, classical communication is not allowed in this setting. Intuitively, CLO allows the parties to exploit pre-established quantum correlations encoded in the catalyst, but since the catalyst is returned in the same state, it cannot be consumed as a resource. We investigate which state transformations are possible in CLO, and in particular highlight its fundamental difference with the LOCC class. We also investigate how CLO relates to LOCC transformations assisted by bounded quantum communication.

More formally, we define a class of operations CLO$^{(d)}$, where local operations are augmented with the use of a catalyst featuring quantum correlations of bounded dimension, i.e., assuming its Schmidt number is upper bounded by a parameter $d$. First, we note that there are transformations in $\text{LOCC}$ that cannot be implemented using $\text{CLO}^{(d)}$. Indeed, classical communication cannot be replaced by a catalyst in general. Then we compare $\text{CLO}^{(d)}$ with the class $\text{SLOCCQ}^{(d)}$ (stochastic local operations with classical and bounded quantum communication) of protocols that allow stochastic local operations, classical communication, and bounded quantum communication, such that the product of the dimensions of the transmitted quantum systems is at most $d$. Intuitively, sending a quantum system of dimension $d$ is at least as powerful as preparing a catalyst with Schmidt number $d$, so $\text{CLO}^{(d)}$ is a subset of $\text{SLOCCQ}^{(d)}$ by construction. Surprisingly, however, catalytic local operations can sometimes outperform SLOCCQ with a respectively smaller budget for quantum communication, showing strict separations between the classes, that is, $\text{CLO}^{(d)} \not\subset \text{SLOCCQ}^{(d-1)}$. These results show that quantum catalysis, classical communication, and bounded quantum communication are incomparable resources and capture genuinely distinct aspects of bipartite quantum state manipulation.

To illustrate the physical relevance of the operation class $\text{CLO}^{(d)}$, consider the following operational scenario. Two distant parties, Alice and Bob, share a correlated quantum system in a known state, on which they want to perform a desired transformation to reach a final target state. To do so, each of them contacts the local branch of a global organization (the entanglement bank). The latter lends them the requested correlated catalyst that may be used for a limited time, but must ultimately be returned in an unaltered state. We show that this scenario allows Alice and Bob to realize certain deterministic state transformations which would be impossible with bounded quantum communication.

\section{Definitions and preliminary results}

We investigate different transformation classes from bipartite quantum states $\rho_{AB} \in \mathcal{L(\mathcal{H}_A\otimes \mathcal{H}_B)}$ to bipartite quantum states of potentially different systems $\sigma_{A'B'}\in \mathcal{L}(\mathcal{H}_{A'}\otimes \mathcal{H}_{B'})$. 
A given transformation class $\tt T$, associates with each input state $\rho_{AB}$ and given output systems (Hilbert spaces $\mathcal{H}_{A'}$ and $\mathcal{H}_{B'})$ a set of density operators
to which that state can be transformed, namely
\begin{equation} \label{eq: state transform}
    {\tt T}_{A'B'} : \rho_{AB} \mapsto  {\tt T}_{A'B'}(\rho_{AB})\subset \mathcal{L}(\mathcal{H}_{A'}\otimes \mathcal{H}_{B'}).
\end{equation}
When the output systems are not specified, we write ${\tt T}(\rho_{AB})$ to denote the union of reachable sets of output states for all dimensions of $A'$ and $B'$. 

Since the world is not limited to the two systems $A(A')$ and $B(B')$, we naturally include in the study of transformation classes input state $\rho_{ABR}\in \mathcal{L}(\mathcal{H}_A\otimes \mathcal{H}_B \otimes \mathcal{H}_R)$ that may share quantum correlations with an auxiliary system $R$, to which Alice and Bob do not have access. In this case the reachable set is ${\tt T}_{A'B'}(\rho_{ABR})\subset \mathcal{L}(\mathcal{H}_{A'}\otimes \mathcal{H}_{B'} \otimes \mathcal{H}_{R})$. To keep the notation simple, we omit $R$ where it is trivial or irrelevant for the discussion. 

In quantum information it is common that a transformation class can be described irrespectively of the state $\rho_{AB}$ on which it acts, i.e., by directly constraining the set of liner CP maps ${\tt T}_{A'B'}(\cdot)$. For instance, this is the case of local operations and shared randomness (LOSR~\cite{schmid2020type}), local operations and shared entanglement (LOSE~\cite{schmid2021postquantum}), local operations and classical communication (LOCC~\cite{Chitambar_2014}), and stochastic local operations, bounded-quantum and classical communication (SLOCCQ in Def.~\ref{def: SLOCCQ} just below). In contrast, this will not be the case for catalytic local transformations (CLO in Def.~\ref{def: clo}), central to this paper, hence the reader shall keep the more general notion of Eq.~\eqref{eq: state transform} in mind. Let us now formally define the later two classes of transformations. 

\begin{definition}[{\rm SLOCCQ}$^{(d)}$] \label{def: SLOCCQ}
    The class of transformations {\rm SLOCCQ}$^{(d)}$ contains all sequences of stochastic local operations, classical communication, and bounded quantum communication. The quantum communication is bounded in the following sense. Let $d_i$ be the dimension of the quantum system sent in the round $i$ from one party to the other ($d_i=1$ if the round only involves classical communication or no communication). We require that
    \begin{equation}
        \prod_i d_i \leq d.
    \end{equation}
\end{definition}

To introduce the class of CLO along the same lines, we want to be able to bound the dimension of the catalyst system in a meaningful way. To do so we recall the definition of the Schmidt number.

\begin{definition} [{\rm Schmidt number~\cite{terhal2000,sanpera2001}}] The Schmidt number of a bipartite state $\rho_{AB}$, is given by 
\begin{align}
    {\rm SN}(\rho_{AB}) \coloneqq \inf_{\{\Psi^{(i)}_{AB}\}} \!\! \max_i &\quad{\rm Rank}(\tr_{A} \Psi^{(i)}_{AB}),
\end{align}
such that $\rho_{AB}=\sum_{i} p_i \Psi_{AB}^{(i)}$, with $p_i\geq0$, is a pure state decomposition of $\rho_{AB}$. 
\end{definition}

The Schmidt number is a discrete measure of bipartite entanglement
quantifying its dimensionality. Let us consider the following definition.

\begin{definition}[{\rm CLO}$^{(d)}$, Fig.~\ref{fig:clodef}] \label{def: clo} The class of bounded catalytic local operations contains all transformations of a bipartite state $\rho_{AB}$ into $\sigma_{A'B'}$ that can be implemented with local operations (CPTP maps) $\mathcal{A}: AC_A \to A'C_A$, $\mathcal{B}:BC_B \to B'C_B$ and a catalyst with a bounded Schmidt number ${\rm SN}(\omega_{C_A C_B})\leq d$, such that 
\begin{align}
\begin{split}
    \tau_{A'B' C_A C_B} &= \mathcal{A}\otimes \mathcal{B} [\rho_{AB}\otimes \omega_{C_A C_B}] 
    \quad \text{with} \\
    \sigma_{A' B'} &= \tr_{C_A C_B}  (\tau_{A'B' C_A C_B}) \quad \text{and} \\
    \omega_{C_A C_B} &= \tr_{A' B'} (\tau_{A'B' C_A C_B}).
\end{split}    
\end{align}
\end{definition}

\begin{figure}
    \centering
    \includegraphics[width=0.75\columnwidth]{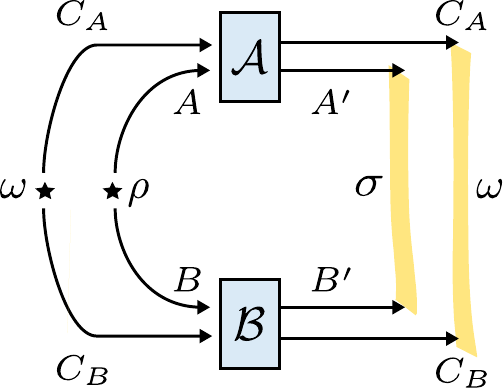}
    \caption{Illustration of catalytic local operations, as given in Definition~\ref{def: clo}. A general state transformation $\rho_{AB}\mapsto \sigma_{A'B'}$, belongs to the class CLO , if it can be realized in the above diagram, while keeping the marginal state of the catalyst $\omega_{C_A C_B}$ unaltered. It belongs to the class {\rm CLO}$^{(d)}$, if, in addition, this state has the Schmidt number bounded by $d$.}
    \label{fig:clodef}
\end{figure}

We will simply write CLO, if $d$ is not specified. Loosely speaking, the constraint ${\rm SN}(\omega_{C_AC_B})\leq d$ guarantees that the catalyst can be encoded by a system with local quantum dimension $d$ (plus any classical degrees of freedom). It is worth mentioning that in the case of state transformations via catalytic LOCC, restrictions on the dimension and other properties of the catalyst have also been considered, see, e.g., Refs.~\cite{eisert2000catalysis,grabowecky2019bounds,rubboli2022fundamental}.  We emphasise that a combination $\{\mathcal{A},\mathcal{B},\omega_{C_AC_B}\}$ defining a CLO transformation for the input state $\rho_{AB}$, is in general not CLO for another input state $\rho_{AB}'$, as the marginal state of the catalyst may not be preserved. This means, in particular, that catalytic transformations cannot be performed in a black-box fashion --- knowledge of the initial state is required in general. Finally, we emphasize that we require the state of the catalyst to be exactly preserved at the end of the transformation, in contrast to quantum state embezzlement~\cite{vanDam2003}; in the latter, this is only the case for a catalyst of unbounded dimension.

\section{Results}
We now state our main results, which describe the nontrivial relationships among the different classes of bipartite transformations studied in this paper. See Fig.~\ref{fig:clod} for a schematic summary.

Let us start from three observations that directly follow from the way we construct the set of transformations ${\rm CLO}^{(d)}$.

\begin{observation}
  {\rm CLO}$^{(d)}\subseteq$ {\rm LOSE}$^{(d)}\subseteq$ {\rm SLOCCQ}$^{(d)}$. 
\end{observation}
  
The first inclusion holds by definition. Here  it is implicit that in {\rm LOSE}$^{(d)}$, $d$ also bounds the Schmidt number of the shared entangled state, which does not need to be returned. For the second inclusion note that any state with Schmidt number $d$ can be prepared by sending a quantum system of dimension $d$ from Alice to Bob, along with classical communication.

\begin{observation} 
    {\rm CLO}$^{(1)}=$ {\rm LOSR}. 
\end{observation}

On the one hand, shared randomness can be viewed as a catalyst---the correlations built with the systems do not change its marginal state. On the other, any separable catalyst can be prepared from shared randomness with local operations.

This observation implies that separable states are free in CLO$^{(1)}$. Hence, any correlation that can be prepared solely using CLO$^{(1)}$ is, by construction, classical. At the same time, any correlations that cannot be prepared with CLO$^{(1)}$ must be genuinely quantum. This provides CLO$^{(1)}$ with a natural operational interpretation: in a manner analogous to LOCC, they allow one to single out classical from quantum correlations.

\begin{observation}
    {\rm LOCC} $\not\subset$  {\rm CLO}$^{(d)}$.
\end{observation}
    
Consider the following classical states
$\rho_{ABR} = \sum_{x=0}^1 \frac{1}{2} [0]_A \otimes [x]_B \otimes [x]_R$ and $\sigma_{ABR} = \sum_{x=0}^1 \frac{1}{2} [x]_A \otimes [0]_B \otimes [x]_R$. Here and below we use the notation $[x]\coloneqq\ketbra{x}$ for classical registers.
It is immediate to see that one bit of classical communication from Bob to Alice allows the transformation
\begin{align}\nonumber
\rho_{ABR}\;\xrightarrow{\;\text{LOCC}\;}\;\sigma_{ABR}.
\end{align}
Furthermore, this bit of communication is  {\it necessary} for Alice and Bob to do the transformation, which is thus impossible with LOSE and hence also with CLO. To see this formally, let $\omega_{C_A C_B}$ be the auxiliary entanglement shared between Alice and Bob and assume that the LOSE transformation is possible. Then, the marginal state $\rho_{AC_A R}= ([0] \otimes \omega)_{AC_A} \otimes (\sum_x \frac{1}{2}[x])_R$ can be transformed to $\sigma_{AR}= \sum_x \frac{1}{2} [x]_A \otimes[x]_R$ by means of a local operation $\mathcal{A}:AC_A\to A$ applied by Alice. This is manifestly a contradiction, since it violates the data processing inequality for the quantum relative entropy $H(R|AC_A)$.

\vspace{0.5cm}

A more interesting and arguably  less intuitive fact is that there exist state transformations in ${\rm CLO}^{(d)}$ that cannot be realized with ${\rm SLOCCQ}^{(d-1)}$. This is the main result of this work.

\begin{theorem}\label{thm}
    {\rm CLO}$^{(d)} \not\subset$ {\rm SLOCCQ}$^{(d-1)}$  with $d = 2^n$.
\end{theorem}

To demonstrate this result, we will make use of the two following lemmas. The first shows the possibility of certain transformations within ${\rm CLO}^{(d)}$, and was established in Ref.~\cite{bavaresco2025catalytic} in the context of catalytic activation of Bell nonlocality, inspired by the catalytic LOCC transformation introduced in Ref. \cite{duan2005multiple}. 

\begin{lemma} \label{lem: cata}
{\rm CLO}$^{(d)}$ with $d=2^{n-1}$ allows any state $\rho_{AB}$ to be transformed into
\begin{equation}\label{eq: out catalyst n}
    \tau_{A' B'} =  \frac{1}{n}\, \rho_{AB}^{\otimes n} + \frac{n-1}{n}\, \sigma_{AB}^{\otimes n}\, ,
\end{equation}
where $A' \equiv A_1\dots A_n$ and $B' \equiv B_1\dots B_n$ are composed of $n\in \mathbb{N}$ copies of the system $A$ and $B$, and where $\sigma_{AB}=\sigma_A \otimes \sigma_B$ is an arbitrary product state, using a catalyst given by
\begin{equation}
    \omega_{C_A C_B} = \frac{1}{n}  \sum_{i = 0}^{n-1} \rho_{AB}^{\otimes i} \ot \sigma_{AB}^{\otimes (n-1-i)} \, ,
\end{equation}
where $C_A\equiv A_1\dots A_{n-1}$ and $C_B\equiv B_1\dots B_{n-1}$ are composed of $n-1$ copies of the system $A$ and $B$.
\end{lemma}

We refer to Ref.~\cite{bavaresco2025catalytic} for the proof of the lemma above.

The second lemma guarantees that the Schmidt number of the initial state cannot change arbitrarily with ${\rm SLOCCQ}^{(d)}$ operations.

\begin{lemma} \label{lem: SN}
  A state $\tau_{A'B'}$ cannot be prepared form $\rho_{AB}$ via {\rm SLOCCQ}$^{(d)}$ if 
\begin{align}
    d< \frac{{\rm SN}(\tau_{A'B'})}{{\rm SN}(\rho_{AB})}.
\end{align}
If $\rho_{AB}$ is pure this is an if and only if.
\end{lemma}
\begin{proof} We start with the ``if'' part. Let us consider a general {\rm SLOCCQ}$^{(d)}$ transformation that acts as
\begin{align}
    \rho_{AB} \rightarrow \tau_{A'B'} = \sum_{i} p(i) \tau_{A'B'}^{(i)},
\end{align}
where $p(i)$ is the probability of branch $i$ and $\tau_{A'B'}^{(i)}$ is the (normalized) post-measurement state. Let us observe that if $\rho_{AB}$ has Schmidt number ${\rm SN}(\rho_{AB}) = k$, then for any branch $i$ of the protocol, the post-measurement state $\tau_{A'B'}^{(i)}$ satisfies ${\rm SN}(\tau_{A'B'}^{(i)}) \leq k$ and thus the Schmidt number of $\tau_{A'B'}$ will be at most ${\rm SN}(\tau_{A'B'}) \leq k$. Notice now that sending a quantum system of dimension $d_i$ can increase the Schmidt number of the state by at most
    \begin{equation}
       { \rm SN}(\tau_{AB}^{(i+1)}) \leq  d_i \cdot { \rm SN}(\tau_{AB}^{(i)}),
    \end{equation}
in all branches of the protocol. This proves the ``if'' direction of the lemma.

    To prove the ``only if'' part, notice that if $\rho_{AB}$ is pure, it can be filtered to the maximally entangled state in dimension ${\rm SN}(\rho_{AB})$ using SLO. Then by sending a quantum system of dimension $d$, Alice and Bob establish a maximally entangled state of dimension $d' =d \cdot {\rm SN}(\rho_{AB})$. From there any state $\tau_{A'B'}$ with ${\rm SN}(\tau_{A'B'})\leq d'$ can be prepared with LOCC. 
\end{proof}

We are now ready to prove Theorem~\ref{thm}.

\begin{proof} {\it of Theorem~\ref{thm}.}
    Let $\rho_{AB}= \ketbra{\Phi^+}$ with $\ket{\Phi^+}=\frac{1}{\sqrt 2} (\ket{00}+\ket{11})$ be the maximally entangled two qubit state and let $\sigma_{AB}=\ketbra{22}_{AB}$. By Lemma~\ref{lem: cata}, $\rho_{AB}$ can be transformed to 
    \begin{equation}
        \label{eq:tau_prop1}
        \tau_{A'B'} = \frac{1}{n+1} \ketbra{\Phi^+}^{\otimes(n+1)}  +\frac{n}{n+1} \ketbra{22}^{\otimes(n+1)}
    \end{equation}
    via CLO$^{(2^{n})}$. We have 
    \begin{align}
        {\rm SN}(\rho_{AB})&=2 \\
        {\rm SN}(\tau_{A'B'})&=2^{n+1}.
    \end{align}
    On the other hand, by Lemma~\ref{lem: SN} $\rho_{AB}$ cannot be transformed to  $\tau_{A'B'}$ with {\rm SLOCCQ}$^{(d)}$ if 
    \begin{equation}
        d < \frac{  {\rm SN}(\tau_{A'B'})}{{\rm SN}(\rho_{AB})} = 2^n,
    \end{equation}
    and $2^n-1<2^n$. Consequently, the state from Eq. \eqref{eq:tau_prop1} cannot be realised from $\rho_{AB}$ via ${\rm SLOCCQ}^{(2^n-1)}$.
\end{proof}
\vspace{0.5 cm}

To show the separation between CLO$^{(d)}$ and SLOCC$^{(d-1)}$, we used the fact that the former class allows to increase the Schmidt number of a state by a factor $d$. This is also the maximal possible increase saturating the bounds $\text{SN}(\tau_{A'B'})\leq \text{SN}(\rho_{AB}\otimes\omega_{C_AC_B})\leq \text{SN}(\rho_{AB}) \times d$.
This raises the natural question of which entanglement measures remain monotonic in presence of a catalyst, or otherwise how much can they increase. The former question was studied in the context of catalytic LOCC transformations. In particular, it was shown that a pure states transformation $\ket{\Psi}_{AB}\to \ket{\Phi}_{AB}$ is possible iff the (von Neumann) entanglement entropy is non-increasing, i.e., if $H(\tr_A \Psi)\geq H(\tr_A \Phi)$~\cite{kondra2021catalytic}. Hence, for pure states entanglement measures inequivalent to entanglement entropy can be increased with CLOCC. In the case of mixed states it is known~\cite{kondra2021catalytic}, that the squashed entanglement~\cite{christandl2004squashed} remains a monotone, since it is additive for product states $E_{sq}(\rho_{AB}\otimes \omega_{C_AC_B})=E_{sq}(\rho_{AB})+E_{sq}(\omega_{C_AC_B})$ and super-additive in general $E(\tau_{A'B'C_A C_B})\geq  E(\tau_{A'B'})+E( \tau_{C_AC_B}=\omega_{C_AC_B})$. Clearly, this conclusion remains true for CLO. Conversely, it is an interesting open question to understand which measures of entanglement, other than the Schmidt number, cease to be monotonic under CLO$^{(d)}$, and how does $d$, or other restrictions on the catalyst, limit their increase.

\section{Conclusions}
Here, we showed that quantum catalysis and dimension-bounded quantum communication are incomparable resources for quantum state transformations. On the one hand, we showed that classical communication cannot in general be substituted by a quantum catalyst, as there are quantum state transformations that are achieved via LOCC and that cannot be reproduced by CLO; in fact, in our demonstrating example, even shared entanglement (i.e. LOSE) is not sufficient to reproduce the effects of classical communication. On the other hand, we prove that a dimension-bounded catalyst can outperform even dimension-bounded quantum communication with a lower quantum-system-dimension budget. As seen in our proof, this effect can be attributed to the ability of transformations in CLO$^{(d)}$ to increase the Schmidt number of the target state without consuming the entanglement in the catalyst. Hence, we conclude that these different classes of transformations are equipped with incomparable capacities when it comes to the task of transforming quantum states.

Our results contribute to the broader program of operational frameworks in quantum information theory. We show that catalysis can play a strong role as a resource enabling state transformations beyond what can be achieved with unlimited classical communication and limited quantum communication. Future work directions include further exploring and quantifying the trade-off between these resources, and investigating further implications of catalytic local transformations for the resource theories of quantum correlations.
\\

\noindent\textit{Acknowledgments.} We are grateful to Marco T\'ulio Quintino for discussions. The authors acknowledge funding from the Swiss National Science Foundation (SNSF) through NCCR SwissMAP (project~182902), project 219366 and the Swiss Postdoctoral Fellowship (project~216979). P.L.-B. also acknowledges funding from Polish National Agency for Academic Exchange (NAWA) through grant BPN/PPO/2023/1/00018/U/00001.


%

\end{document}